\documentclass[submission,copyright,creativecommons]{eptcs}
\providecommand{\event}{NCMA 2024} 

\usepackage{graphicx}
\usepackage{mathtools}
\usepackage{todonotes}
\usepackage{amssymb}
\usepackage{bm}
\usepackage{algorithm}
\usepackage{algorithmicx}
\usepackage{algpseudocode}
\usepackage{float}
\usetikzlibrary{backgrounds}

\usepackage{amsthm}
\theoremstyle{plain}
\newtheorem{theorem}{Theorem}
\newtheorem{lemma}{Lemma}
\newtheorem{observation}{Observation}
\newtheorem{corollary}{Corollary}
\theoremstyle{definition}
\newtheorem{definition}{Definition}
\newtheorem{example}{Example}
\theoremstyle{remark}
\newtheorem*{conjecture}{Conjecture}

\newtheorem{property}{Property}
\newtheorem{remark}{Remark}

\usepackage{iftex}

\ifpdf
  \usepackage{underscore}         
  \usepackage[T1]{fontenc}        
\else
  \usepackage{breakurl}           
\fi

\newcommand{\Nat}{\mathbb{N}}
\newcommand{\Int}{\mathbb{Z}}
\newcommand{\pos}{\mathrm{pos}}
\newcommand{\Run}{\mathrm{Run}}

\newcommand{\A}{\mathcal{A}}
\newcommand{\Q}{{}Q{}}

\renewcommand{\L}{\mathcal{L}}

\renewcommand{\rho}{\varrho}

\newcommand{\reset}{\mathord{\circlearrowleft}}
\newcommand{\inc}{\mathrm{inc}}
\newcommand{\dec}{\mathrm{dec}}
\newcommand{\id}{\textsc{id}}

\newcommand*\rk{\mathrm{rk}}
\newcommand*\size{\mathrm{size}}
\newcommand*\height{\mathrm{ht}}
\newcommand*\sub{\mathrm{sub}}
\newcommand*{\gentree}{\mathrm{tree}}
\newcommand*{\stateseq}{\mathrm{stateseq}}
\newcommand*{\statepos}{\mathrm{statepos}}

\newcommand{\compA}[1]{\mathrm{comp}_\A(#1)}
\newcommand*\lex{\leq_{\mathrm{lex}}}

\newcommand\xqed[1]{%
  \leavevmode\unskip\penalty9999 \hbox{}\nobreak\hfill
  \quad\hbox{#1}}
\newcommand{\exqed}{\xqed{$\triangleleft$}}

\DeclarePairedDelimiter\abs{\lvert}{\rvert}%

\tikzset{
  inlinesf/.style={
    baseline={([yshift=-0.6ex]current bounding box.center)}
  , level distance=20pt
  , text height=1.5ex
  , text depth=0.75ex
  , sibling distance=20pt
  , inner sep=1.5pt
  }
}

\title{Non-Global Parikh Tree Automata}
\author{Luisa Herrmann
\institute{Computational Logic Group, TU Dresden, Germany}
\institute{ScaDS.AI Center for Scalable Data Analytics and Artificial Intelligence\\ Dresden/Leipzig, Germany}
\email{luisa.herrmann@tu-dresden.de}
\and Johannes Osterholzer \institute{secunet Security Networks AG, Germany}\email{johannes.osterholzer@gmail.com}
}
\def\titlerunning{Non-Global Parikh Tree Automata}
\def\authorrunning{L. Herrmann, J. Osterholzer}
\begin{document}
\maketitle

\begin{abstract}
Parikh (tree) automata are an expressive and yet computationally well-behaved extension of finite automata -- they allow to increment a number of counters during their computations, which are finally tested by a semilinear constraint. In this work, we introduce and investigate a new perspective on Parikh tree automata (PTA): instead of testing one counter configuration that results from the whole input tree, we implement a non-global automaton model. Here, we copy and distribute the current configuration at each node to all its children, incrementing the counters pathwise, and check the arithmetic constraint at each leaf. We obtain that the classes of tree languages recognizable by global PTA and non-global PTA are incomparable. In contrast to global PTA, the non-emptiness problem is undecidable for non-global PTA if we allow the automata to work with at least three counters, whereas the membership problem stays decidable. However, for a restriction of the model, where counter configurations are passed in a linear fashion to at most one child node, we can prove decidability of the non-emptiness problem.
\end{abstract}

\section{Introduction}

Finite automata are one of the most fundamental computation models in theoretical computer science and have been generalized to many structures that go beyond words, such as trees \cite{FetWri68}. However, they are not sufficient when arithmetic properties (such as two symbols occurring equally often) have to be ensured. For this reason, there are numerous approaches to extending automata with counting mechanisms. In the area of tree automata, although less studied than extensions of word automata, there are (among others) two approaches that come into consideration: On the one hand, there are \emph{pushdown tree automata} \cite{Gue81}, which extend pushdown automata to trees and thus recognize context-free tree languages, as well as their restriction \emph{counter tree automata}. And on the other hand, Parikh tree automata \cite{KlaRue02,Kla04} have been considered: during their computations, they allow to increment a number of counters in each step. These counters are finally tested to satisfy a semilinear constraint. The calculation principles for the counting mechanisms work orthogonally in both approaches -- while pushdown tree automata split their computations at each node and execute them pathwise, Parikh tree automata allow a global view: their counters are incremented over the whole input tree before their membership in a semilinear set is tested. Thus, in the remaining work we will refer to this model as \emph{global Parikh tree automata} (GPTA).

One motivation for the development and investigation of Parikh automata in recent years is the specification and verification of systems that fall outside the scope of regular languages. For such applications, tree automata are also interesting, as they are more suitable to model non-determinism and parallel processes than word automata. However, we think that a non-global view would be interesting for this case in particular: with GPTA, requirements such as "the same arithmetic property applies in every alternative path" cannot be modeled.

For this reason and inspired by the computation strategy of pushdown tree automata, we introduce here an alternative, non-global definition of Parikh tree automata (PTA): we define a model which copies and distributes the current counter configuration at each node to all its children, thus increments the counters pathwise, and finally checks at each leaf node whether the obtained configuration is contained in a semilinear set. In this way, PTAs are able to test arithmetic properties of tree paths.

\paragraph{Contributions} In this work, we start the investigation of non-global PTA, especially from the perspective of their expressiveness and decidability:
\begin{itemize}
    \item We generalize GPTA, which have so far only been considered for complete binary trees, to trees over arbitrary ranked alphabets and prove an exchange lemma (Lemma \ref{gpta-exchange}): This lemma, originally shown for Parikh word automata \cite[Lemma 1]{CadFM11}, states that certain parts in computations of GPTA can be rearranged. Thus, it can be used to show the limits of their expressive power.
    \item The exchange lemma is used to show that there are languages that are recognized by non-global PTA but not by GPTA. The converse of this statement is also shown, and thus we obtain that the language classes of these two models are incomparable (Theorem \ref{thm:uncomp}).
    \item We prove that, in contrast to GPTA, non-emptiness is undecidable for PTA with at least three counters (Theorem \ref{thm:undecidable}). This follows from a simulation of the computations of two counter machines. On the other hand, membership is decidable for arbitrary PTA (Theorem \ref{thm:member}).
    \item We introduce a restriction on the computation mechanism for PTA: \emph{linear} PTA may at each node only pass the current counter configuration to one child tree, the computations in all other child trees start again with all counters zero. With this restriction, we can limit the number of such "non-reset" paths that must at least occur if the language of a linear PTA is non-empty (Lemma \ref{lem:spinal-bounded}). Thereby, we can show that non-emptiness becomes decidable (Theorem \ref{thm:lin-empt}).
\end{itemize}

\paragraph{Related work} Since their introduction in \cite{KlaRue02,KlaRue03}, Parikh automata have been studied in many works from a variety of perspectives, cf.\ for example \cite{CadFM11,CadGhoPer23}; recently there have also been extensions for infinite words \cite{GuhaJL022,Gro24} and infinite trees \cite{HPR24}. It is known that Parikh automata correspond to a special form of \emph{vector addition systems with states} (VASS) over integers, so-called $\Int$-VASS \cite{HaaHal14}. For VASS, in \cite{CouSch14} a definition for alternation was provided by using branching -- in this sense, PTA could be seen as a formulation of \emph{alternating $\Int$-VASS}.

The idea of limiting the counter flow in a computation to linear paths originates from \emph{linear pushdown tree automata} \cite{FujiyoshiK00} and has later also been extended to \emph{tree automata with storage} \cite{Her21}.

\section{Preliminaries}

We denote by $\Nat$ the set of \emph{natural numbers} including 0 and set $[n]=\{1,\ldots,n\}$ for each $n\in\Nat$. Given a finite set $A$, we denote its \emph{cardinality}, i.e., the number of its elements, by $|A|$. For each $k\in \Nat$ and $a_1,\ldots,a_k\in A$, we call $w=a_1\ldots a_k$ a \emph{word over $A$} and say that its \emph{length} is $k$. We let $A^k$ be the set of all words over $A$ of length $k$ and set $A^*=\bigcup_{n\in \Nat}A^k$. Given some $w\in A^*$, we refer to its length by $|w|$ and the word of length 0 will be denoted by $\varepsilon$. We let $\sqsubseteq$ denote the \emph{prefix order} on $A$: for words $w_1,w_2\in A$ we have $w_1\sqsubseteq w_2$ if $w_2 = w_1u$ for some $u\in A^*$. The \emph{lexicographic order} on $\Nat^*$ will be denoted by $\lex$, and is defined for every \(u\), \(v \in \Nat^*\) such that \(u \lex v\) whenever \emph{(i)}~either \(u \sqsubseteq v\), or \emph{(ii)}~there are \(x\), \(y\), \(z \in \Nat^*\) and \(n\), \(m \in \Nat\) such that \(u = xny\), \(v = xmz\) and \(n < m\).

\paragraph{\textbf{Alphabets, trees, and tree languages.}} A \emph{ranked set} is a tuple $(\Sigma,\rk)$ where $\Sigma$ is a set (its elements called \emph{symbols} or \emph{labels}) and $\rk\colon\Sigma\to\Nat$ is a function assigning to each symbol in $\Sigma$ a natural number, its \emph{rank}. We often assume $\rk$ implicitly and only write $\Sigma$ instead of $(\Sigma,\rk)$.
For each $n\in\Nat$, by $\Sigma^{(n)}$ we mean $\rk^{-1}(n)$ and we write $\sigma^{(n)}$ in order to say that $\sigma\in\Sigma^{(n)}$.
We say that a ranked set \((\Sigma,\rk)\) is a \emph{ranked alphabet} if the set \(\Sigma\) is finite.\footnote{Most tree languages in this paper will have labels from some finite ranked alphabet. However, we have to allow infinite label sets for one definition.}

Now let $\Sigma$ be a ranked set and $H$ a set. The set $T_\Sigma(H)$ of \emph{trees} (over $\Sigma$ and indexed by $H$) is defined to be the smallest set $T$ such that \emph{(i)} $H\subseteq T$ and \emph{(ii)} for each $n\in\Nat$, $\sigma\in\Sigma^{(n)}$, and $\xi_1,\ldots,\xi_n\in T$ we have $\sigma(\xi_1,\ldots,\xi_n)\in T$. If $H=\emptyset$, we simply write $T_\Sigma$ instead of $T_\Sigma(H)$. As usual, we denote the tree $\alpha()$ by $\alpha$ for each $\alpha\in\Sigma^{(0)}$ and we often write monadic trees of the form $\gamma_1(\gamma_2(\ldots\gamma_k(\#)\ldots)$, where $\gamma_1,\ldots,\gamma_k\in\Sigma^{(1)}$, as words $\gamma_1\ldots\gamma_k\#$. Each subset $L\subseteq T_\Sigma$ is called a \emph{tree language}.

Let $\xi$, $\zeta\in T_\Sigma(H)$. We let $\pos(\xi)\subseteq\Nat^*$ denote the \emph{set of positions} of $\xi$, defined in the usual way: for every $\alpha\in\Sigma^{(0)}\cup H$ we let $\pos(\alpha)=\{\varepsilon\}$ and for every $n\geq 1$, $\sigma\in\Sigma^{(n)}$, and $\xi_1,\ldots,\xi_n\in T_\Sigma(H)$ we let $\pos(\sigma(\xi_1,\ldots,\xi_n))=\{\varepsilon\}\cup\{i\rho\mid i\in[n],\rho\in\pos(\xi_i)\}$. Furthermore, $|\xi|=|\pos(\xi)|$ stands for the \emph{size} of $\xi$, and, given a position $\rho\in\pos(\xi)$, we denote by $\xi(\rho)$ the \emph{label of $\xi$ at position $\rho$} and by $\xi_{|\rho}$ the \emph{subtree of $\xi$ at position $\rho$}, respectively.
Let \(\xi[\zeta]_\rho\) designate the tree that results from \(\xi\) by replacing the subtree rooted at \(\rho\) by \(\zeta\).
We let $\height(\xi) = \max \{ |\rho| \mid \rho \in \pos(\xi) \}$ and $\sub(\xi)=\{\xi_{|\rho}\mid\rho\in\pos(\xi)\}$. Given positions $\rho_1,\rho_2\in\pos(\xi)$ we say that the subtrees $\xi_{|\rho_1}$  and $\xi_{|\rho_2}$ are \emph{independent} if $\rho_1\not\sqsubseteq\rho_2$ and $\rho_2\not\sqsubseteq\rho_1$.

A \emph{path} $\pi$ is a sequence $\pi=\rho_1\ldots\rho_n$ of positions $\rho_1,\ldots,\rho_n\in\pos(\xi)$ such that for each $i\in[n-1]$ we have $\rho_{i+1}=\rho_i k$ for some $k\in [\rk(\xi(\rho_i))]$. The \emph{path word} of $\pi$ is given by $\xi(\pi)=\xi(\rho_1)\ldots\xi(\rho_n)$. We say that $\pi$ is a \emph{complete path} (or c-path) if $\rho_1=\varepsilon$ and $\xi(\rho_n)\in\Sigma^{(0)}$; the set of all complete paths of $\xi$ is denoted by $\mathrm{paths}(\xi)$.

\begin{example}\label{Ex1}
Consider the ranked alphabet $\Sigma=\{\sigma^{(2)},\gamma^{1},\alpha^{(0)}\}$ as well as the set $H=\{u\}$. Then $\sigma(\alpha,\alpha)$ is a tree in $T_\Sigma$ and $\xi=\sigma(\sigma(\gamma(\alpha), \alpha),\gamma(\gamma(u)))$ is a tree in $T_\Sigma(H)$.
As mentioned above, we will also sometimes write $\sigma(\sigma(\gamma\alpha, \alpha),\gamma\gamma u)$ for \(\xi\).
We have $\pos(\xi)=\{\varepsilon,1,11,111,12, 2,21,211\}$, $\xi(\varepsilon)=\sigma$ and $\xi(21)=\gamma$, and $\xi_{|2}=\gamma(\gamma(u))$.
The trees in this example can be represented graphically as
\[
\begin{tikzpicture}[inlinesf,sibling distance=2.5em]
  \node (t) {\(\sigma\)}
  child[] { node {\(\alpha\)} }
  child[] { node {\(\alpha\)}};
 \end{tikzpicture}
\qquad\qquad\text{and}\qquad\qquad
\begin{tikzpicture}[inlinesf,sibling distance=2.5em]
  \node (t) {\(\sigma\)}
  child[] { node {\(\sigma\)}
    child { node {\(\gamma\)}
      child { node {\(\alpha\)}}}
    child { node {\(\alpha\)}}}
  child[] { node {\(\gamma\)}
    child { node {\(\gamma\)}
    child { node {\(u\)}}}
 };
 \end{tikzpicture}
\quad\text{,} 
 \]
respectively.\exqed
\end{example}

\paragraph{\textbf{Contexts, spines, and composition.}}
Let $X=\{x_1,x_2,\ldots\}$ be a fixed set of \emph{variables} that is disjoint from every other set in this work and let $X_n=\{x_1,\ldots,x_n\}$. Now let $H$ be a set, $k\geq 1$ and $\xi\in T_\Sigma(H\cup X_k)$. We call $\xi$ a \emph{context} if \emph{(a)} there is exactly one $\rho_i\in\pos(\xi)$ (in the further denoted by $\pos_{x_i}(\xi)$) with $\xi(\rho_i)=x_i$ for each $i\in[k]$ and \emph{(b)} for each $i_1,i_2\in[k]$, if $i_1<i_2$, then $\pos_{x_{i_1}}(\xi)\leq_{\mathrm{lex}}\pos_{x_{i_2}}(\xi)$. The set of all such \emph{contexts over $\Sigma$ and $H$} will be denoted by $C_\Sigma(H,X_k)$ (or by $C_\Sigma(X_k)$ if $H=\emptyset$).

The \emph{composition} $\zeta\cdot\xi$ of a context $\zeta\in C_\Sigma(H,X_1)$ and a tree $\xi\in T_\Sigma(H\cup X)$ replaces $x_1$ in $\zeta$ by $\xi$. This operation can be transferred to arbitrary $k\geq 2$, trees $\xi\in T_\Sigma(H\cup X_k)$ and $\xi_1,\ldots,\xi_k\in T_\Sigma(H)$: we let $\xi[\xi_1,\ldots,\xi_k]$ stand for the tree $\zeta$ obtained from $\xi$ by replacing each occurrence of $x_i$ by $\xi_i$ for each $i\in[k]$.

Given a tree $\xi\in T_\Sigma$ and a path $\rho_1\ldots\rho_n$, we let the \emph{$(\rho_1,\rho_n)$-spine of $\xi$}, denoted by $\xi^{[\rho_1,\rho_n]}$, be the context $\zeta\in C_\Sigma(X_k)$ for some $k\in\Nat$ containing exactly the path $\rho_1\ldots\rho_n$, i.e., such that there are a context $\zeta'\in C_\Sigma(X_1)$ and trees $\xi_1,\ldots,\xi_k\in T_\Sigma$ with $\xi=\zeta'\cdot(\zeta[\xi_1,\ldots,\xi_k])$ and for each $\rho\in\Nat^*$ we have $\rho\in\pos(\zeta)\setminus\{\pos_{x_i}(\zeta)\mid i\in[k]\}$ if and only if $\rho_1\rho=\rho_j$ for some $j\in[n]$.

\begin{example}
Consider the positions $\rho_1=1$ and $\rho_2=11$ of the tree $\xi$ from Example~\ref{Ex1}. Then the $(\rho_1,\rho_2)$-spine of $\xi$ is $\xi^{[\rho_1,\rho_2]}=\sigma(\gamma(x_1),x_2)$, and we have
\(\xi=\sigma(x_1,\gamma \gamma u)\cdot(\sigma(\gamma(x_1),x_2) [\alpha, \alpha])\). \exqed
\end{example}

\paragraph{\textbf{Semilinear sets.}} 
Let $s\geq 1$. We denote by $\bm{0}_s$ the zero-vector $\bm{0}_s=(0,...,0)\in\Nat^s$ of \emph{dimension} $s$. If $s$ is clear from the context, we often only write $\bm{0}$.{}
A set $C\subseteq\Nat^s$, $s\geq 1$, is \emph{linear} if it is of the form
$C=\{d_0 + \textstyle\sum_{i\in[l]}m_i d_i \mid m_1,\ldots,m_l\in\Nat\}$
for some $l\in\Nat$ and vectors $d_0,\ldots,d_l\in\Nat^s$. Any finite union of linear sets is called \emph{semilinear}.

\begin{lemma}[{\cite[Theorem 1.2 and 1.3]{GinSpa66}}]\label{lemma:dec-C}
    Given a semilinear set $C\subseteq\Nat^s$ and a vector $d\in\Nat^s$, it is decidable whether $d\in C$.
\end{lemma}

\paragraph{Parikh string automata}

Let $s\geq 1$. A \emph{Parikh string automaton} of dimension $s$ ($s$-PA) is a tuple $\A=(Q,\Sigma,q_0,\Delta,F,C)$ where $Q$ is a finite set of states, $\Sigma$ is a (string) alphabet, $q_0\in Q$ (initial state), $F\subseteq Q$ (final states), $\Delta$ is a finite set of transitions of the form $(q,a,d,q')$ for $q,q'\in Q$, $a\in\Sigma$, $d\in\Nat^s$, and $C$ is a semi-linear set over $\Nat^s$. Let $w\in\Sigma^*$. A \emph{run} of $\A$ on $w$ is a sequence \[(p_0,a_1,d_1,p_1)(p_1,a_2,d_2,p_2)\ldots(p_{n-1},a_n,d_n,p_n)\]of transitions such that $p_0=q_0$, $a_1\ldots a_n=w$, $p_n\in F$, and $(d_1+\ldots+d_n)\in C$. The set of all runs of $\A$ on $w$ is denoted $\mathrm{Run}_\A(w)$ and the \emph{language recognized by $\A$} is the set $\L(\A)=\{w\in\Sigma^*\mid \mathrm{Run}_\A(w)\neq\emptyset\}$.

\begin{lemma}[{\cite[Property 6]{KlaRue03}}]\label{pa-dec}
    Given a Parikh string automaton $\A$, it is decidable whether $\L(\A)\neq\emptyset$.
\end{lemma}

\section{Global Parikh Tree Automata}

Let us recall the definition of (global) Parikh tree automata from \cite{KlaRue02,Kla04}. Note that we use here a slight variation of the original version: In \cite{KlaRue02,Kla04}, only full binary trees were considered and, thus, the number of successor states in each transition was fixed to two (also for leaf nodes). In this paper, we extend Parikh tree automata to arbitrary ranked trees in the usual way -- it is not hard to see that for alphabets containing only binary and nullary symbols, both formalisms are equivalent.

\paragraph{\textbf{Extended Parikh map.}} 
Given a ranked alphabet $\Sigma$ and some finite $D\subseteq\Nat^s$ for $s\geq 1$, the automaton model works with symbols from $\Sigma\times D$. Thus, we use the \emph{projections} 
$\cdot_\Sigma:{\Sigma\times D}\to \Sigma$ with $(a,d)_\Sigma=a$ and
$\cdot_D:{\Sigma\times D}\to D$ with $(a,d)_D=d$, extended to trees in the obvious way.  
Moreover, the \emph{extended Parikh map} $\Psi\colon T_{\Sigma\times D}\to\Nat^s$ is defined for each tree $\xi\in T_{\Sigma\times D}$ by
$\Psi(\xi)=\sum_{\rho\in\pos(\xi)}(\xi(\rho))_D\,.$ 

\paragraph{\textbf{Global Parikh tree automata.}} Let $m\geq 1$. A \emph{global Parikh tree automaton} of dimension $m$ ($m$-GPTA) is a tuple $\A=(Q,\Sigma, D,q_0,\Delta,C)$ where $Q$ is a finite set of states, $\Sigma$ is a ranked alphabet, $D\subset \Nat^m$ is finite, $q_0\in Q$ is the initial state, $C\subseteq \Nat^m$ is a semilinear set, and $\Delta$ is a finite set of transitions of the form
\[q\to \langle\sigma,d\rangle(q_1,\ldots,q_n)\]
where $n\in\Nat$, $\sigma\in\Sigma^{(n)}$, $d\in D$, and $q,q_1,\ldots,q_n\in Q$.

Given a tree $\zeta\in T_{\Sigma\times D}$, a \emph{run of $\A$ on $\zeta$} is a mapping $r\colon\pos(\zeta)\to Q$ such that for each $\rho\in\pos(\zeta)$, $r(\rho)\to \zeta(\rho)(r(\rho1),\ldots,r(\rho n))$ with  $n=\rk(\zeta(\rho))$ is in $\Delta$. We say that a run $r$ is \emph{successful} if $r(\varepsilon)=q_0$ and $\Psi(\zeta)\in C$; we denote the set of all successful runs of $\A$ on $\zeta$ by $\Run_\A(\zeta)$. Then the \emph{language of $\A$}, denoted by $\L(\A)$ is the set $\L(\A)=\{\xi \in T_\Sigma\mid \exists \zeta\in T_{\Sigma\times D} \text{ with } (\zeta)_\Sigma=\xi \text{ and }\Run_\A(\zeta)\neq\emptyset\}$.

\subsection{Pumping-style Exchange Lemma for GPTA}

For Parikh automata, a classical pumping lemma that cuts out or iterates parts of a computation is not known -- missing or additional parts in a computation would change the extended Parikh image and thus affect acceptance. However, since the final counter configuration is a global result of the entire computation, parts of the computation can be rearranged without changing the extended Parikh image. This was shown in \cite[Lemma 1]{CadFM11} for the string case and is generalized here to the tree case. This result will be useful later to show that certain tree languages recognizable by non-global Parikh tree automata are not GPTA-recognizable. Note that a crucial part of the extension is that computation parts from independent subtrees are reordered. This allows us to distinguish path counting from global counting using the exchange lemma. Figure \ref{fig:exchange} is a graphical representation of the following lemma.

\begin{figure}
    \centering
    \includegraphics[width=\textwidth]{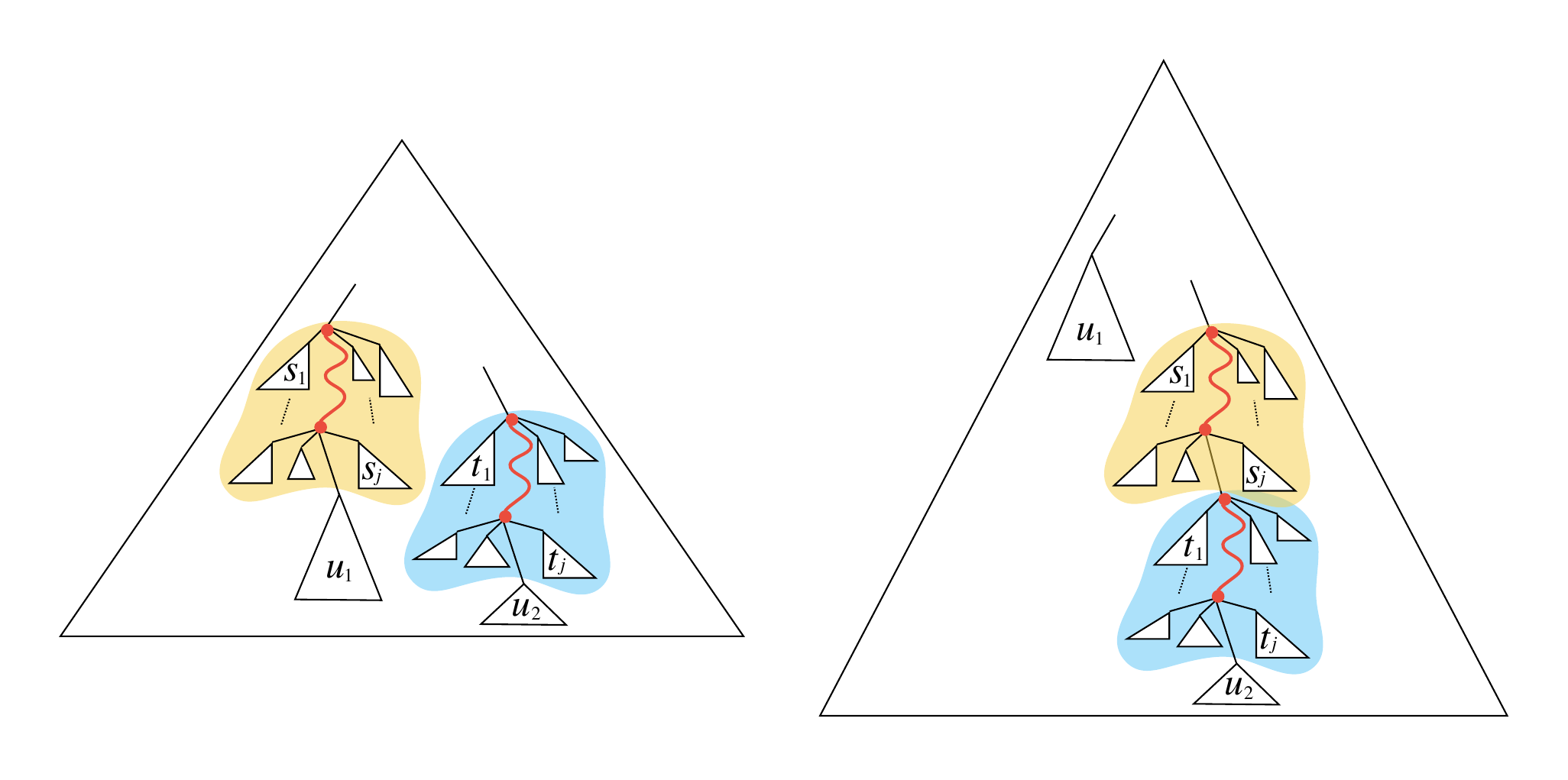}
    \caption{The tree $\xi$ divided as in Lemma \ref{gpta-exchange} (1.) and its reordering as in (2.) where the red spine corresponds to $\zeta_2$.}
    \label{fig:exchange}
\end{figure}

\begin{lemma}\label{gpta-exchange}
    Let $L$ be a GPTA-recognizable tree language. Then there exist constants $l,p>0$ such that for each tree $\xi \in L$ with at least $l$ pairwise independent subtrees of height at least $p$ there exists $k\geq 0$, contexts $\zeta_1\in C_\Sigma(X_{2}),\zeta_2\in C_\Sigma(X_{k+1})$ with $0<\height(\zeta_2)< p$, trees $s_1,\ldots,s_k$, $t_1,\ldots,t_k$, $u_1,u_2\in T_\Sigma$, and $j\in[k+1]$ such that
    \begin{enumerate}
        \item $\xi=\zeta_1[\zeta_2[s_1,\ldots, s_{j-1}, x_1,s_{j},\ldots,s_k]\cdot u_1,\zeta_2[t_1,\ldots, t_{j-1}, x_1,t_{j},\ldots,t_k]\cdot u_2]$,
        \item $\zeta_1[u_1,\zeta_2[s_1,\ldots, s_{j-1}, x_1,s_{j},\ldots,s_k]\cdot \zeta_2[t_1,\ldots, t_{j-1}, x_1,t_{j},\ldots,t_k]\cdot u_2]\in L$, and
        \item $\zeta_1[\zeta_2[s_1,\ldots, s_{j-1}, x_1,s_{j},\ldots,s_k]\cdot \zeta_2[t_1,\ldots, t_{j-1}, x_1,t_{j},\ldots,t_k]\cdot u_1,u_2]\in L$.
    \end{enumerate}
\end{lemma}

\begin{proof}
    Let $\A=(Q,\Sigma, D,q_0,\Delta,C)$ be a GPTA with $\L(\A)=L$ and let $p=|Q|+1$. Further, let $N$ be the maximal rank of symbols in $\Sigma$.  In order to define $l$, we build from the transitions of $\A$ a graph $G$ labeled by elements of $\Delta\times N$ as follows: We let $G=(V,E)$ with $V=Q$ and $E\subseteq Q\times (\Delta\times [N]) \times Q$ such that $(q,\langle\tau,i\rangle,q')\in E$ if and only if $\tau=q\to \langle\sigma,d\rangle(q_1,\ldots,q_n)$, $i\leq n$, and $q_i=q'$. Now let $l'$ be the number of cycles in $G$, i.e., the number of sequences $(f_0,u_1,f_1)(f_1,u_2,f_2)...(f_{k-1},u_k,f_k)$ for $k\leq p$ such that $(f_{i-1},u_i,f_i)\in E$ for each $i\in[k]$, $f_0=f_k$, and there are no $i,j\in[k]$ such that $i\ne j$ and $f_i=f_j$. Then $l=l'+1$.

    Now consider $\xi\in L$ that fulfills the requirements of the statement. Then there exists a tree $t\in T_{\Sigma\times D}$ with $(t)_\Sigma=\xi$ and a successful run $r\in\Run_\A(t)$. By our requirement for $\xi$, also $t$ contains $l$ independent subtrees of height at least $p$. Note that because of their height, each of these subtrees contains a cycle.  By the definition of \(l\), we can apply the pigeonhole principle and obtain that there is a pair of paths that contain the same cycle: there has to be some $1\leq h\leq p$ and two paths $\rho^1_1\ldots\rho^1_h$ and $\rho^2_1\ldots\rho^2_h$ in independent subtrees of $t$ that induce transition cycles which coincide. Formally, for each $i\in[2]$, let
    \[w_i=(r(\rho^i_1),\tilde\tau(\rho^i_1),r(\rho^i_2))\ldots(r(\rho^i_{h-1}),\tilde\tau(\rho^i_{h-1}),r(\rho^i_{h})),\]
    where $\tilde\tau(\rho^i_j)=\langle r(\rho^i_j)\to t(\rho^i_j)(r(\rho^i_j1),\ldots,r(\rho^i_jn)),\mu\rangle$, and \(\mu\in\Nat\) such that $\rho^i_{j+1}=\rho^i_j\mu$. Both $w_1$ and $w_2$ are cycles in $G$, thus $r(\rho^i_1)=r(\rho^i_{h})$, and $w_1=w_2$.


    Now let $\zeta_1\in C_{\Sigma}(X_2)$ such that $\xi=\zeta_1[\xi_{|\rho^1_1},\xi_{|\rho^2_1}]$ and let $\zeta_2=\xi^{[\rho^{1}_1,\rho^{1}_{h-1}]}=\xi^{[\rho^{2}_1,\rho^{2}_{h-1}]}$. Clearly, there is some $k\in\Nat$, $j\in[k+1]$, and trees $s_1,\ldots,s_k,t_1,\ldots,t_k\in T_{\Sigma}$ such that $\xi_{|\rho^1_1},\xi_{|\rho^2_1}$ can be written as
    \[\xi_{|\rho^1_1}=\zeta_2[s_1,\ldots, s_{j-1}, x_1,s_{j},\ldots,s_k]\cdot \xi_{|\rho^{1}_h} \qquad \text{ and } \qquad \xi_{|\rho^2_1}=\zeta_2[t_1,\ldots, t_{j-1}, x_1,t_{j},\ldots,t_k]\cdot \xi_{|\rho^{2}_h}\,.\]
    By letting $u_1=\xi_{|\rho^{1}_h}$ and $u_2=\xi_{|\rho^{2}_h}$ we obtain
    \[\xi=\zeta_1[\zeta_2[s_1,\ldots, s_{j-1}, x_1,s_{j},\ldots,s_k]\cdot u_1,\zeta_2[t_1,\ldots, t_{j-1}, x_1,t_{j},\ldots,t_k]\cdot u_2]\]
    which corresponds to item (1.) of the statement. Note that we can subdivide $t$ in exactly the same building blocks as $\xi$: there are $\delta_1\in C_{\Sigma\times D}(X_2)$, $\delta_2=t^{[\rho^{1}_1,\rho^{1}_{h-1}]}$, and $\hat s_1,\ldots,\hat s_k,\hat t_1,\ldots,\hat t_k$ such that $\pos(\delta_1)=\pos(\zeta_1)$, $\pos(\hat s_i)=\pos(s_i)$, $\pos(\hat t_i)=\pos(t_i)$ for each $i\in[k]$, and we have
    \[t=\delta_1[\delta_2[\hat s_1,\ldots, \hat s_{j-1}, x_1,\hat s_{j},\ldots,\hat s_k]\cdot t_{|\rho^{1}_h},\delta_2[\hat t_1,\ldots, \hat t_{j-1}, x_1,\hat t_{j},\ldots,\hat t_k]\cdot t_{|\rho^{2}_h}]\,.\]

    For item (2.) we need to argue that the reordering \[\xi'=\zeta_1[u_1,\zeta_2[s_1,\ldots, s_{j-1}, x_1,s_{j},\ldots,s_k]\cdot \zeta_2[t_1,\ldots, t_{j-1}, x_1,t_{j},\ldots,t_k]\cdot u_2]\] of $\xi$ can be recognized by $\A$, too. To show this, we construct from $r$ a computation $r'$ on the corresponding reordering $t'$ of $t$ given by 
    \[t'=\delta_1[t_{|\rho^{1}_h},\delta_2[\hat s_1,\ldots, \hat s_{j-1}, x_1,\hat s_{j},\ldots,\hat s_k]\cdot\delta_2[\hat t_1,\ldots, \hat t_{j-1}, x_1,\hat t_{j},\ldots,\hat t_k]\cdot t_{|\rho^{2}_h}]\] 
    as follows:
    \begin{itemize}
    \item for each $\rho\in\pos(\delta_1)\setminus\{\pos_{x_1}(\delta_1),\pos_{x_2}(\delta_1)\}$ we let $r'(\rho)=r(\rho)$,
    \item for each $\rho\in\pos(t_{|\rho^{1}_h})$ we let $r'(\rho^{1}_1\rho)=r(\rho^{1}_h\rho)$,
    \item for each $\rho\in\pos(\delta_2[\hat s_1,\ldots, \hat s_{j-1}, x_1,\hat s_{j},\ldots,\hat s_k])\setminus\{\pos_{x_1}(\delta_2[\hat s_1,\ldots, \hat s_{j-1}, x_1,\hat s_{j},\ldots,\hat s_k])\}$ we let\\ $r'(\rho^{2}_1\rho)=r(\rho^{1}_1\rho)$, and
    \item for all $\rho\in\pos(\delta_2[\hat t_1,\ldots, \hat t_{j-1}, x_1,\hat t_{j},\ldots,\hat t_k]\cdot t_{|\rho^{2}_h})$ we let $r'(\rho^{2}_h\rho)=r(\rho^{2}_1\rho)$.
    \end{itemize}
    It remains to argue that $r'$ is successful on $t'$. But this is easy to see: as we only cut out and inserted a part of the tree at positions which carry the same state, all transitions are still applicable. Finally, as $\Psi(t')=\Psi(t)$, we obtain $r'\in\Run_\A(t')$ and, thus, $\xi'\in\L(\A)$. 

    The proof of item (3.) is analogous.
\end{proof}

\section{Non-Global Parikh Tree Automata}

Now we define a non-global variant of Parikh tree automata in which not an extended Parikh image of a whole input tree is computed, but counter vectors that occur in computations \emph{(i)} are added up pathwise and \emph{(ii)} it is tested at each leaf node whether the resulting counter configuration is contained in the semilinear set $C$ of the automaton.

Here we do not consider counter vectors as additional labelings that we guess beforehand, but use them in the transitions as operations which can differ per path, similar as it is done the case of pushdown tree automata. Therefore, a transition that reads a $k$-ary symbol can send $k$ different vectors to the different subtrees. Additionally, we allow a reset operation $\reset$ that sets each counter configuration back to $\bm{0}$. This operation will be needed later to define Parikh tree automata that pass the current counter configuration of each node to exactly one child node instead of copying it to all children. Such a reset operation has also been introduced in the context of tree automata with storage to define a linear model \cite{Her21} and was used for an extension of Parikh string automata (over infinite words) \cite{Gro24}.

\paragraph{\textbf{Non-global Parikh tree automata.}} Let $m\geq 1$. A \emph{(non-global) Parikh tree automaton of dimension $m$ with reset operation} ($m$-PTAR) is a tuple $\A=(Q,\Sigma,q_0,\Delta,C)$ where $Q$ is a finite set of states, $\Sigma$ is a ranked alphabet, $q_0\in Q$ is the initial state, $C\subseteq \Nat^m$ is a semilinear set, and $\Delta$ is a finite set of transitions of the form
\[q\to \sigma(q_1(d_1),\ldots, q_n(d_n))\]
where $n\in\Nat$, $\sigma\in\Sigma^{(n)}$, $q,q_1,\ldots,q_n\in Q$, and $d_1\ldots d_n\in(\Nat^m\cup\{\reset\})$.


The semantics of an $m$-PTAR $\A=(Q,\Sigma,q_0,\Delta,C)$ is defined as follows. We denote by $\textsc{ID}$ the set $Q\times \Nat^m$ of \emph{automaton configurations}, each consisting of a state and a \emph{counter configuration} from $\Nat^m$. For each transition $\tau\in\Delta$ we let $\Rightarrow^\tau$ be the binary relation on the set $T_\Sigma(\textsc{ID})$ such that for each $\zeta_1,\zeta_2\in T_\Sigma(\textsc{ID})$ we have
\[\zeta_1\Rightarrow^\tau\zeta_2\]
if there are $\hat\zeta\in C_\Sigma(\textsc{ID},X_1)$, $\hat\zeta_1,\hat\zeta_2\in T_\Sigma(\textsc{ID})$ such that $\zeta_1=\hat\zeta\cdot\hat\zeta_1$, $\zeta_2=\hat\zeta\cdot\hat\zeta_2$, and either
\begin{itemize}
    \item $\tau=q\to \sigma(q_1(d_1),\ldots, q_n(d_n))$ for some $n\geq 1$, $\hat\zeta_1=(q,w)$, and $\hat\zeta_2=\sigma((q_1,w_1),\ldots,(q_n,w_n))$ where, for each $i\in[n]$, $w_i=w+ d_i$ if $d_i\neq\reset$ and $w_i=\bm{0}$ otherwise, or
    \item $\tau=q\to\alpha$ for some $\alpha\in\Sigma^{(0)}$, $\hat\zeta_1=(q,w)$, $w\in C$, and $\hat\zeta_2=\alpha$.
\end{itemize}
The \emph{computation relation of $\A$} is the binary relation $\Rightarrow_\A=\bigcup_{\tau\in\Delta}\Rightarrow^\tau$. A \emph{computation} is a sequence $t=\zeta_0\Rightarrow^{\tau_1}\zeta_1\ldots\Rightarrow^{\tau_n}\zeta_n$ (sometimes abbreviated as $\zeta_0\Rightarrow^{\tau_1\ldots\tau_n}\zeta_n$) such that $n\in\Nat$, $\zeta_0,\ldots,\zeta_n\in T_\Sigma(\textsc{ID})$, $\tau_1,\ldots,\tau_n\in \Delta$, and $\zeta_{i-1}\Rightarrow^{\tau_i}\zeta_i$ for each $i\in[n]$. We say that the $\emph{length}$ of $t$ is $n$. We call $t$ \emph{successful on $\xi\in T_\Sigma$} if $\zeta_0=(q_0,\bm{0})$ and $\zeta_n=\xi$; the set of all successful computations of $\A$ on $\xi$ is denoted by $\compA{\xi}$. The \emph{language recognized by $\A$} is the set $\L(\A)=\{\xi\in T_\Sigma\mid \compA{\xi}\neq\emptyset\}$.

\paragraph{} Now we consider an example showing the capability of non-global Parikh tree automata to check a semi-linear property for each path in a tree.

\begin{example}\label{L-ab}
    Let $\Sigma=\{a^{(2)},b^{(2)},\#^{(0)}\}$. We consider the tree language $L_{ab}$ containing all trees $\xi$ such that the word of labels of each complete path in $\xi$ is of the form $a^n b^n\#$ for some $n\geq 1$, i.e.,
    \[L_{ab}=\{\xi\in T_{\Sigma}\mid \forall \pi\in\mathrm{paths}(\xi): \xi(\pi)\in\{a^nb^n\#\mid n\geq 1\}\}\,.\]

    This tree language can by recognized by the
    2-PTA $\A=(Q,\Sigma,q_a,\Delta,C)$ where $Q=\{q_a,q_b\}$, $C=\{(i,i)\mid i\geq 1\}$, and $\Delta$ contains the transitions
    \[\tau_{a,1}\colon\ q_a\to a(q_a(1,0),q_a(1,0))\qquad \tau_{a,2}\colon\ q_a\to b(q_b(0,1),q_b(0,1))\]
    and 
    \[\tau_{b,1}\colon\ q_b\to b(q_b(0,1),q_b(0,1))\qquad \tau_{b,2}\colon\ q_b\to \#\ .\]
    The intuition behind this automaton is quite easy: for each $a$ it reads, the first counter component is increased by $1$ and for each $b$ it reads, the second counter component is increased by $1$. Finally, $\#$ can only be computed if the value of the first and the second component is equal. This process becomes clear if we look at a part of some computation for $a(b(\#,\#),b(\#,\#))\in L_{ab}$: let us consider a computation of the form
    \begin{align*}
    (q_a,(0,0))&\Rightarrow^{\tau_{a,1}}a\bigl((q_a,(1,0)),(q_a,(1,0))\bigr)\\&\Rightarrow^{\tau_{a,2}}a\bigl(b((q_b,(1,1)),(q_b,(1,1))),(q_a,(1,0))\bigr)\\
    &\Rightarrow^{\tau_{b,2}}a\bigl(b(\#,(q_b,(1,1))),(q_a,(1,0))\bigr)\Rightarrow^* a(b(\#,\#),b(\#,\#))\,.
    \end{align*}
    Observe that in the third step of the computation, from the automaton configuration $(q_b,(1,1)))$ the application of $\tau_{b,2}$ to compute the leaf $\#$ is allowed only because $(1,1)\in C$.
    \exqed
\end{example}

\begin{remark}\label{remark}
We note that the tree language $L_{ab}$ is a \emph{context-free tree language}: it is a simple observation that the set of path words occurring in its trees is context-free and, thus,   $L_{ab}$ can be recognized by a pushdown tree automaton. However, we can easily extend Example \ref{L-ab} to paths of the form $a^nb^nc^n\#$ by using a third counter -- the resulting tree language would not be context-free anymore.
\end{remark}

\subsection{Restrictions of PTAR}

If $\reset$ does not occur in the transitions of $\A$, we call $\A$ an $m$-PTA. Moreover, we say that $\A$ is \emph{linear} if for each transition $q\to \sigma(q_1(d_1),\ldots, q_n(d_n))$ in $\Delta$ there is at most one $i\in[n]$ with $d_i\in\Nat^m$ and for all $j\neq i$ we have $d_j=\reset$, i.e., at each node the storage is either completely reset or passed to exactly one child.

\paragraph{Spinal computation trees}
Let us define an alternative semantics for linear PTAR, needed later when we prove decidability of their non-emptiness problem.
The idea is to recursively structure the computations of such a linear PTAR \(\A\) as follows: during the computation on an input tree \(\xi\), we always apply the rewrite relation \(\Rightarrow_\A\) to the node \(w \in \pos(\xi)\) that has been passed the storage from its parent node, if there is any such node.  When there is no longer such a node, we apply this process recursively to the remaining nodes labeled by states.

\begin{example}
\label{ex:spinal-comput}
For an example, consider the linear \(1\)-PTAR \(\A = (Q, \Sigma, q, \Delta, C)\), where
\(Q = \{q, p\}\), \(\Sigma = \{\sigma^{(2)}, \alpha^{0}\}\), \(C = \Nat\), and \(\Delta\) contains the transitions
\[q \to \sigma\bigl(q(\reset), q(1)\bigr)\,\text{,} \qquad q \to \sigma\bigl(q(2), p(\reset)\bigr)\,\text{,} \qquad q \to \alpha\,\text{,} \qquad p \to \sigma(q(\reset), q(4)\,\text{,}\]
denoted by \(\tau_1\), \(\tau_2\), \(\tau_3\), and \(\tau_4\), respectively.  Assume the following computation of \(\A\).
\[
\begin{tikzpicture}[inlinesf,sibling distance=2.5em]
\node (t) {\((q, 0)\)};
 \begin{scope}[on background layer]
   \draw[line width=3ex,red!50,line cap=round,line join=round] (t.center) +(-2mm,0) -- (t.center) -- +(2mm,0);
 \end{scope}
\end{tikzpicture}
\Rightarrow^{\tau_1}
\begin{tikzpicture}[inlinesf,sibling distance=2.5em]
  \node (t) {\(\sigma\)}
  child[] { node {\((q,0)\)} }
  child[] { node {\((q,1)\)} };
 \begin{scope}[on background layer]
   \draw[line width=3.75ex,red!50,line cap=round,line join=round] (t.center) to (t-2.center);
 \end{scope}
\end{tikzpicture}
\Rightarrow^{\tau_2}
\begin{tikzpicture}[inlinesf,sibling distance=2.5em]
  \node (t) {\(\sigma\)}
  child[] { node {\((q,0)\)} }
  child[] { node {\(\sigma\)}
    child { node {\((q,3)\)}}
    child { node {\((p,0)\)}}
 };

 \begin{scope}[on background layer]
   \draw[line width=3.75ex,red!50,line cap=round,line join=round] (t.center) to (t-2.center) to (t-2-1.center);
 \end{scope}
\end{tikzpicture}
\Rightarrow^{\tau_3}
\begin{tikzpicture}[inlinesf,sibling distance=2.5em]
  \node (t) {\(\sigma\)}
  child[] { node {\((q,0)\)} }
  child[] { node {\(\sigma\)}
    child { node {\(\alpha\)}}
    child { node {\((p,0)\)}}
 };

 \begin{scope}[on background layer]
   \draw[line width=3.75ex,red!50,line cap=round,line join=round] (t.center) to (t-2.center) to (t-2-1.center);
 \end{scope}
\end{tikzpicture}
\Rightarrow^*
\begin{tikzpicture}[inlinesf,sibling distance=2.5em]
  \node (t) {\(\sigma\)}
  child[] { node[xshift=-2.5ex] {\(\sigma\)}
    child { node {\(\alpha\)}}
    child { node {\(\sigma\)}
      child { node {\(\alpha\)}}
      child { node {\(\alpha\)}}}}
  child[] { node[xshift=2.5ex] {\(\sigma\)}
    child { node {\(\alpha\)}}
    child { node {\(\sigma\)}
      child { node {\(\alpha\)}}
      child { node {\(\alpha\)}}}};

 \begin{scope}[on background layer]
   \draw[line width=3.75ex,red!50,line cap=round,line join=round] (t.center) to (t-2.center) to (t-2-1.center);
   \draw[line width=3.75ex,blue!50,line cap=round,line join=round] (t-1.center) to (t-1-2.center) to (t-1-2-2.center);
   \draw[line width=3.75ex,yellow!50,line cap=round,line join=round] (t-2-2.center) to (t-2-2-2.center);
   \draw[line width=3.75ex,gray!50,line cap=round,line join=round] (t-1-1.center) to (t-1-1.center);
   \draw[line width=3.75ex,green!50,line cap=round,line join=round] (t-1-2-1.center) to (t-1-2-1.center);
   \draw[line width=3.75ex,purple!50,line cap=round,line join=round] (t-2-2-1.center) to (t-2-2-1.center);
 \end{scope}
\end{tikzpicture}
\]

Since in the computation's first step, the storage has been passed to the second child, labeled \((q,1)\), we rewrite this position in the next step, and so on, until the leaf node \(\alpha\) is reached.  The path along which this process takes place is shaded in red.
Afterwards, the process can be applied recursively to the states which did not receive the storage of their parent, resulting in the paths shaded in other colors.

The gist of this section is that the shaded parts of this computation can also be arranged into a tree of subcomputations of the form
\[
\begin{tikzpicture}[inlinesf,sibling distance=4em,
level distance=3em,
  every node/.style={rounded corners=3mm, minimum width=4ex, inner sep=4pt}
  ]
  \node[fill=red!50] {\((q, 0) \Rightarrow^* \sigma((q, 0), \sigma( \alpha, (p, 0)))\)}
  child{ node[xshift=-2cm,fill=blue!50] {\((q, 0) \Rightarrow^* \sigma((q, 0), \sigma( (q, 0), \alpha))\)}
    child { node[xshift=-1cm,fill=gray!50] {\((q, 0) \Rightarrow^* \alpha\)} }
    child { node[xshift=1cm,fill=green!50] {\((q, 0) \Rightarrow^* \alpha\)} }}
  child {node[xshift=2cm,fill=yellow!50] {\((q, 0) \Rightarrow^* \sigma((q, 0), \alpha)\)}
    child { node[fill=purple!50] {\((q, 0) \Rightarrow^* \alpha\)} }};
\end{tikzpicture}
\,\text{,}
\]
called a \emph{spinal computation tree.}  We will prove that if the tree language of a PTAR is not empty, then there is a spinal computation tree of bounded height, leading to a decision procedure.
\exqed
\end{example}

To formally define the notion of spinal computation trees, we have to restrict the derivation relation so that only children that received the storage from their parents can be rewritten.  We do so by constructing a new automaton which only has transitions for such positions.
For this, assume a linear $m$-PTAR \(\A=(Q,\Sigma,q_0,\Delta, C)\).
Let \(Q' = Q \cup \hat Q\), where \(\hat Q = \{\hat q \mid q \in Q\}\).
We construct the linear \(m\)-PTAR \(\A' = (Q', \Sigma, \hat q_0, \Delta', C)\) , where \(\Delta'\) is defined as follows.
\begin{itemize}
\item For every transition of the form \(q \to \alpha\) in \(\Delta\), the set \(\Delta'\) contains the transition \(\hat q \to \alpha\).
\item For every transition \(q \to \sigma(q_1(d_1),\ldots, q_n(d_n))\) in \(\Delta\), where \(d_i = \reset\) for each \(i \in [n]\), the set \(\Delta'\) contains the transition
\[\hat q \to \sigma(q_1(d_1),\ldots, q_n(d_n))\,\text{.}\]
\item Finally, consider a transition \(q \to \sigma(q_1(d_1),\ldots, q_n(d_n))\) in \(\Delta\) such that \(d_i \ne \reset\) for some \(i \in [n]\).  Then the transition
\[\hat q \to \sigma\bigl(q_1(d_1),\ldots, q_{i-1}(d_{i-1}), \hat q_i(d_i),  q_{i+1}(d_{i+1}), \ldots, q_n(d_n)\bigr)\]
is in \(\Delta'\).
\end{itemize}
Note that there are only transitions for states from \(\hat Q\) in \(\Delta'\), the computation cannot continue on states from \(Q\).

Consider a computation
\[\zeta_0 \Rightarrow^{\tau_1} \zeta_1 \Rightarrow^{\tau_2} \cdots \Rightarrow^{\tau_n} \zeta_n\]
of \(\A'\), such that \(n > 0\), \(\zeta_0 = (\hat q, \bm 0)\) for some \(q \in Q\), and \(\zeta_n \in T_\Sigma(Q \times \Nat^m)\).  We call such a computation a \emph{spine computation} of \(\zeta_n\) from \(q\).
In fact, it is easy to see from the definition of \(\A'\) that for every occurrence in \(\zeta_n\) of a tuple \((q, c)\) with \(q \in  Q\) and \(c \in \Nat^m\), it is the case that \(c = \bm 0\).

The set of all spine computations from \(q\) will be denoted by \(S_q\), and given such a spine computation \(s \in S_q\), the generated tree \(\zeta_n\) will be denoted by \(\gentree(s)\).  Moreover, assume that \(\{w_1, \cdots, w_\ell\} \subseteq \pos(\gentree(s))\), for some \(\ell \in \Nat\), is the set of positions in \(\gentree(s)\) with labels from \(Q \times \Nat^m\), and assume that \(w_1\), \ldots, \(w_\ell\) are in left-to-right order, i.e. \(w_1<_{\mathrm{lex}} \cdots <_{\mathrm{lex}} w_\ell\).
Then we will write \(\statepos(s)\) for the sequence \(w_1 \cdots w_\ell\).
Additionally, if for every \(i \in [\ell]\), we have \(\gentree(s)(w_i) = (q_i, \bm 0)\), then we will denote the sequence \(q_1 \cdots q_\ell\) by \(\stateseq(s)\).
For instance, when we write \(s\) for the red-shaded subcomputation from Example~\ref{ex:spinal-comput}, we would have \[\gentree(s) = \sigma((q,0), \sigma(\alpha, (p,0)))\,\text{,}\qquad
\statepos(s) = 1\text{ }\text{~}22\text{,}\qquad\text{and}\qquad
\stateseq(s) = q\,p\text{.}
\]

Now, for every \(q \in Q\), the set of \emph{spinal computation trees} \(D_q\) is defined to be the smallest set such that the following property holds: for every spine computation \(s \in S_q\) with \(\stateseq(s) = q_1\cdots q_\ell\) where \(\ell \in \Nat\), and for every \(d_i \in D_{q_i}\), where \(i \in [\ell]\), the tree \(s(d_1, \ldots, d_\ell)\) is an element of \(D_q\).\footnote{Note that this is the point mentioned in the preliminaries, because of which we must allow trees with labels from an infinite ranked set: the set of labels used for \(D_q\) is the set of spine computations \(\bigcup_{q\in Q}S_q\).}

Given a spinal computation tree \(d = s(d_1, \ldots, d_\ell) \in D_q\) with \(\statepos(s) = w_1\cdots w_\ell\), we finally define the computed tree \(\gentree(d) \in T_\Sigma\) by
\[\gentree(d) = \gentree(s) [\gentree(d_1)]_{w_1} \cdots [\gentree(d_\ell)]_{w_\ell}\,\text{.}\]
This recursive definition is well-behaved because we chose \(D_q\) to be the smallest set of trees fulfilling the given property.

The following lemma relates the rewrite semantics of PTAR to the notion of spinal computation trees.

\begin{lemma}
  \label{lem:semantics}
  Let \(\A=(Q,\Sigma,q_0,\Delta, C)\) be a linear PTAR, let \(q \in Q\), and let \(\xi \in T_\Sigma\).
  Then \((q, \bm 0) \Rightarrow^*_\A \xi\) if and only if there is some \(d \in D_{q}\) with \(\gentree(d) = \xi\).
  In particular, \(\xi \in \L(\A)\) if and only if there is some \(d \in D_{q_0}\) with \(\gentree(d) = \xi\).
\end{lemma}
The direction ``if'' of the lemma can be shown by recursively ``composing'' the spine computations in \(d\).  For the direction ``only if'', one has to reorder the computation of \(\xi\) such that it begins with the computation steps along the spine where no reset operation is performed.  Then these steps correspond to a spine computation \(s\). As all computations besides the spine start in a configuration \((q, \bm 0)\) for some state \(q \in Q\), this process can be repeated recursively to obtain a spinal computation tree \(d\).

\begin{lemma}
  \label{lem:spinal-bounded}
  For every linear PTAR \(\A\) with state set \(Q\), if \(\L(\A) \ne \emptyset\), then there is some spinal computation tree \(d \in D_{q_0}\) such that \(\height(d) \leq \abs{\Q}\).
\end{lemma}
\begin{proof}
  By Lemma~\ref{lem:semantics}, we know that \(\L(\A) \ne \emptyset\) implies the existence of some \(d \in D_{q_0}\). Assume that there is some path \(\rho_1\ldots\rho_n\) in \(d\) such that \(n > \abs{Q}\).  But then there are two distinct indices \(i\) and \(j \in [n]\), say \(i < j\), such that \(d(\rho_i) \in S_q\) and \(d(\rho_j) \in S_q\) for some \(q \in Q\).
Construct
\[d' = d\bigl[ d_{|\rho_j} \bigr]_{\rho_i}\,\text{.}\]
It is easy to see that \(d'\) is also a valid spinal computation tree in \(D_{q_0}\), by inspection of the property used in their definition.  Moreover, \(\size(d') < \size(d)\), so by iterating this construction a finite number of times, we obtain a tree \(d'' \in D_{q_0}\) such that \(\height(d'') \leq \abs{Q}\).
\end{proof}

Now we turn to an example for a tree language that is recognizable by a linear PTAR and still quite powerful: Although each counter configuration is passed to exactly one subtree, this PTAR ensures that an arithmetical constraint holds on each c-path in the trees it accepts.

\begin{example}
    Let $\Sigma=\{a^{(2)},b^{(2)},c^{(1)},d^{(1)},\#^{(0)}\}$. Given a word $w\in\Sigma^*$, we denote by $\mathrm{pref}(w)$ the set of all nonempty prefixes of $w$, i.e., $\mathrm{pref}(w)=\{u\in\Sigma^+\mid u\sqsubseteq w\}$. Now consider the tree language $L_\mathrm{lin}$ consisting of trees $\xi$ of the form \[a(u_1,a(u_2,...a(u_n,b(u_{n+1},...b(u_{2n},\#)))))\] for some $n\geq 1$ and with $u_i\in\{c^md^m\#\mid m\geq 1\}$ for each $i\in[2n]$. Thus, for each $\xi\in L_\mathrm{lin}$ there is exactly one c-path $\pi\in\mathrm{paths}(\xi)$ with $\xi(\pi)=a^nb^n\#$ for some $n\geq 1$ and for each remaining c-path $\pi'\in\mathrm{paths}(\xi)$ we have $\xi(\pi')=wu$ with $w\in\mathrm{pref}(a^nb^n)$ and $u\in\{c^md^m\#\mid m\geq 1\}$.

    The tree language $L_\mathrm{lin}$ can be recognized by the following linear $2$-PTAR: $\A=(Q,\Sigma,q_0,\Delta,C)$ where $Q=\{q_a,q_b,q_c,q_d\}$, $C=\{(i,i)\mid i\geq 1\}$, and $\Delta$ contains the following transitions:
    \begin{itemize}
        \item $q_a\to a(q_c(\reset),q_a(1,0))$, $q_a\to b(q_c(\reset),q_b(0,1))$
        \item $q_b\to b(q_c(\reset),q_b(0,1))$, $q_b\to \#$, and
        \item $q_c\to c(q_c(1,0))$, $q_c\to d(q_d(0,1))$,
        \item $q_d\to d(q_d(0,1))$, $q_d\to \#$
    \end{itemize}
    Thus, in the states $q_a$ and $q_b$, $\A$ counts number of $a$s and $b$s, respectively. When switching into state $q_c$, the counter configuration is reset and from now on it counts the number of $c$s and $d$s.
    \exqed
\end{example}

\section{Expressiveness}

In this section we want to examine how the different formalisms we have introduced relate to each other in terms of their expressiveness.

\subsection{GPTA and PTA}
First, we want to compare the classes of tree languages recognizable by PTA and GPTA. As the different counting mechanisms of the two models already intuitively suggest, the two classes are incomparable. For the formal proof we use tree languages which require counting on paths, or global counting, respectively. We start by showing that the tree language $L_{ab}$ from Example \ref{L-ab} can not be recognized by a global Parikh tree automaton by using the exchange lemma we obtained for GPTA (Lemma \ref{gpta-exchange}).

\begin{lemma}\label{gpta-subset-pta}
The tree language $L_{ab}$ is not GPTA-recognizable.
\end{lemma}

\begin{proof}
Assume that there is some GPTA $\A$ with $\L(\A)=L_{ab}$ and let $p,l\in \Nat$ as in the proof of Lemma \ref{gpta-exchange}. Now consider, for $n\geq 1$, the tree 
\[\xi_n=
    \begin{tikzpicture}[inlinesf,sibling distance=4em] 
      \node {\(a\)}
      child[] { node {\(\textsc{Bin}(1)\)}}
      child[] { node {\(a\)}
      	child { node {\(\textsc{Bin}(2)\)}}
        child { node(n)[yshift=-0.5em, xshift=0.4em]{\rotatebox{-4}{\(\ddots\)}}}};

      \node[right = -0.3em of n, yshift= -1em] {$a$}
      	child[] { node {\(\textsc{Bin}(n)\)}}
      	child[] { node {\(\textsc{Bin}(n)\)}};
    \end{tikzpicture}
 	\] 
where $\textsc{Bin}(n)$ is the complete binary tree over $b$ of height $n$ inductively defined by $\textsc{Bin}(1)=b(\#,\#)$ and $\textsc{Bin}(i)=b(\textsc{Bin}(i-1),\textsc{Bin}(i-1))$ for each $i>1$. Clearly, $\xi_n\in L_{ab}$ for each $n\in\Nat$.

Now choose $n$ big enough such that there are at least $l$ independent subtrees of height at least $p$ in $\xi_n$ and, thus, the requirements of Lemma \ref{gpta-exchange} are fulfilled. However, it is not hard to see that we will not find a context $\zeta_2$ as in Lemma \ref{gpta-exchange} in $\xi_n$ such that item (2.) and (3.) of the lemma are satisfied: as $\zeta_2$ needs to occur in two independent subtrees, it can only consist of $b$s. However, cutting out $b$s in one subtree and inserting them in another one necessarily leads to c-paths that are not of the form $a^k b^k\#$ anymore and, thus, the resulting tree $\xi'$ is not in $L_{ab}$. This is a contradiction.
\end{proof} 

We can use the exchange lemma in a very similar way to show that $L_\mathrm{lin}$, which is recognizable by a linear PTAR, is not GPTA-recognizable either.

\begin{corollary}
  The tree language $L_\mathrm{lin}$ is not GPTA-recognizable.
\end{corollary}

For the other direction, we consider a tree language where the number of symbol occurrences on two different paths are compared. This global counting behavior cannot be implemented by non-global PTA.

\begin{example}
    Let $\Sigma=\{\sigma^{(2)},\gamma^{(1)},\#^{(0)}\}$ and consider the tree language
    \[L_{\gamma\gamma}=\{\sigma(\gamma^n\#,\gamma^n\#)\mid n\in\Nat\}\]
    which can be recognized by the following 2-GPTA $\A$: We let $\A=(Q,\Sigma, D,q_0,\Delta,C)$ where $Q=\{q_0,q_1,q_2\}$, $D=\{(0,0),(1,0),(0,1)\}$, $C=\{(i,i)\mid i\in\Nat\}$, and $\Delta$ consists of the transitions
    \begin{itemize}
        \item $q_0\to\langle\sigma,(0,0)\rangle(q_1,q_2)$,
        \item $q_1\to\langle\gamma,(1,0)\rangle(q_1)$ and $q_2\to\langle\gamma,(0,1)\rangle(q_2)$, as well as
        \item $q_1\to\langle\#,(0,0)\rangle$ and $q_2\to\langle\#,(0,0)\rangle$.\exqed
    \end{itemize}
\end{example}

\begin{lemma}\label{pta-subset-gpta}
    The language $L_{\gamma\gamma}$ is not PTA-recognizable.
\end{lemma}

\begin{proof}
Assume towards a contradiction there is some $m\in\Nat$ and an $m$-PTA $\A$ with $L(\A)=L_{\gamma\gamma}$. As all trees in $L_{\gamma\gamma}$ are of the shape $\sigma(\gamma^n\#,\gamma^n\#)$ for some $n\in\Nat$, each computation of $\A$ on some $\xi=\sigma(\xi_1,\xi_1)\in L_{\gamma\gamma}$ has to be of the form
\[(q_0,\bm{0})\Rightarrow\sigma((q_1,s_1),(q_2,s_2))\Rightarrow_\A^*\xi\]
for some $(q_1,s_1),(q_2,s_2)\in\id$. 
As $\Delta$ is finite, there are only finitely many configurations $(q,s)\in\id$ reachable from $(q_0,\bm{0})$ in one step by reading a $\sigma$ and occurring in a successful computation; we denote the set of all those configurations by $\id_1$. However, as $L_{\gamma\gamma}$ is infinite, there has to be a $(q,s)\in\id_1$ with $(q,s)\Rightarrow_\A^*\gamma^{n_1}\#$, $(q,s)\Rightarrow_\A^*\gamma^{n_2}\#$, and $n_1\neq n_2$.
Suppose without loss of generality that \((q,s)\) is reached in the left subtree of \(\sigma\), i.e.\ \((q_0,\bm{0}) \Rightarrow_\A \sigma((q,s), (q',s'))\) for some \((q', s') \in \id\).

By the assumption on $(q,s)$, there exists a computation
\[(q_0,\bm{0})\Rightarrow\sigma((q,s),(q',s'))\Rightarrow_\A^*\sigma(\xi_1,\xi_1)\] and we can assume that $\xi_1=\gamma^{n_1}\#$. But by the above also 
\[(q_0,\bm{0})\Rightarrow\sigma((q,s),(q',s'))\Rightarrow_\A^*\sigma(\gamma^{n_2}\#,\xi_1)\]
and $\sigma(\gamma^{n_2}\#,\xi_1)\notin L_{\gamma\gamma}$, which is a contradiction.
\end{proof}

From Lemma \ref{gpta-subset-pta} and \ref{pta-subset-gpta} it immediately follows that the tree languages recognizable by GPTA and PTA are incomparable.

\begin{theorem}\label{thm:uncomp}
    The classes of tree languages recognizable by GPTA and PTA are incomparable.
\end{theorem}




    


    

\subsection{PTA, PTAR, and Linear PTAR}

In contrast to the string case, in the tree case a reset cannot be simulated simply by guessing the last reset position: because of branching, a counter configuration could be processed further in one subtree, while a reset takes place in the second subtree. This observation is illustrated by the following example.

\begin{example}
We consider the ranked alphabet $\Sigma=\{\sigma^{(2)},a^{(1)},b^{(1)},\alpha^{(0)}\}$ as well as the $2$-PTAR $\A=(\{q_0,q_1\},\Sigma,q_0,\Delta,C)$ where $C=\{(i,i)\mid i\in\Nat\}$ and $\Delta$ contains the following transitions:
\begin{itemize}
    \item $q_j\to a(q_j(1,0))$ and $q_j\to b(q_j(0,1))$ for $j\in\{0,1\}$,
    \item $q_0\to\sigma(q_1(0,0),q_0(\reset))$, $q_0\to\sigma(q_1(0,0),q_1(0,0))$, and
    \item $q_1\to\alpha$
\end{itemize}
It is easy to observe that for each tree $\xi\in\L(\A)$ it holds that if the context $\sigma(x_1,w_1(\sigma(w_2\alpha,x_2)))$ occurs in $\xi$ for some $w_1,w_2\in\{a,b\}^*$, then the number of $a$s in $w_1w_2$ equals the number of $b$s in $w_1w_2$. Moreover, each $\sigma$ only occurs on the rightmost c-path in $\xi$. Thus, each $\xi\in\L(\A)$ is of the form
    \begin{figure}[H]
        \centering
        \includegraphics[width=19em]{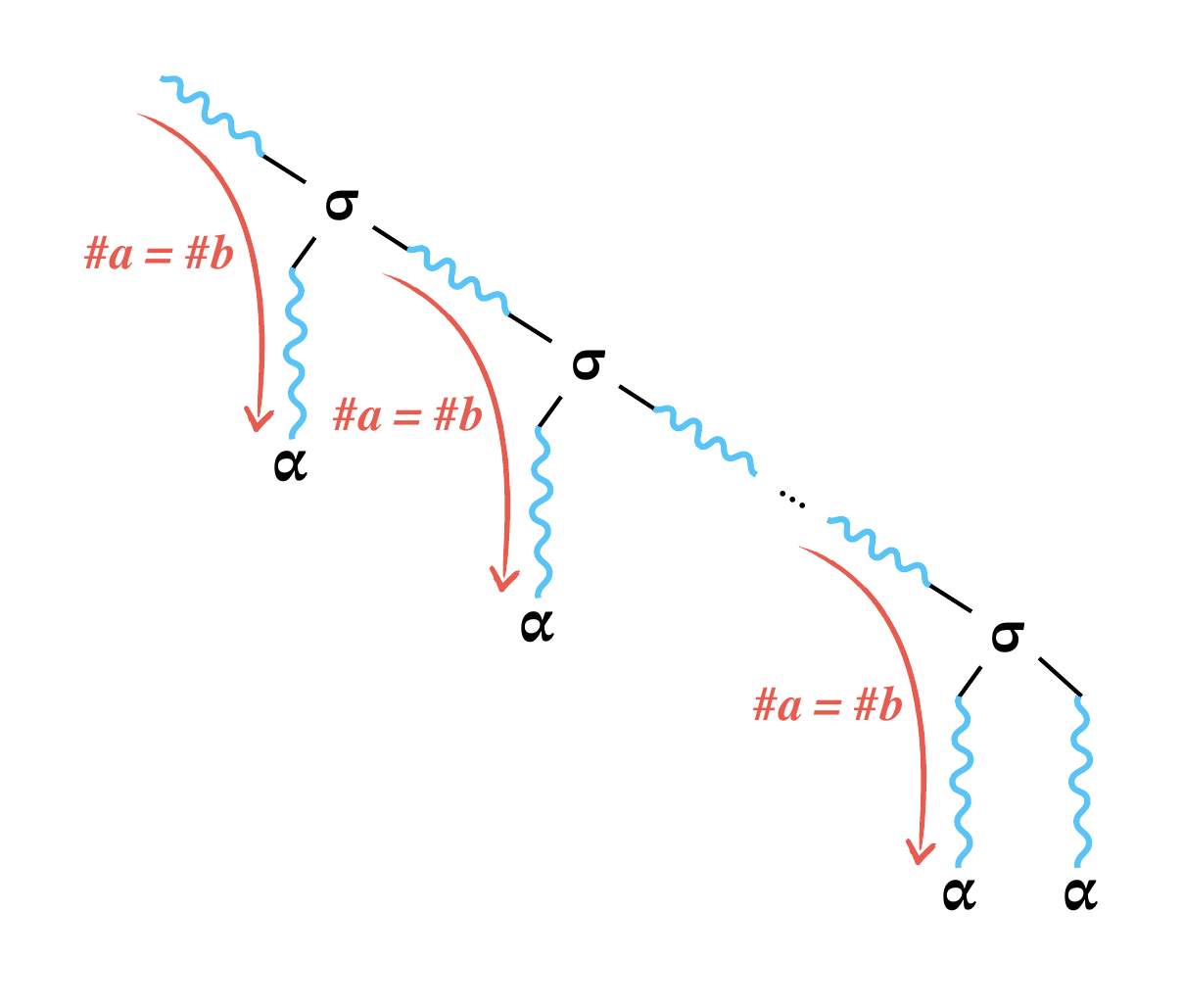}
        \label{fig:enter-label}
    \end{figure}
where the red arrows indicate the paths in $\xi$ on which the number constraint on $a$s and $b$s is tested, respectively. As there might be arbitrary many such tests that are not calculated in completely independent subtrees, it is crucial to reset the counter configuration in between.
\exqed
\end{example}

Therefore, we strongly expect PTAR to be more expressive than PTA. However, our proof methods for PTA were not sufficient to formally prove this statement.

\begin{conjecture}\label{conjec}
    PTA are strictly less expressive than PTAR.
\end{conjecture}

Finally, we observe that also the property of a PTAR to be linear restricts its expressive power.

\begin{lemma}
    Linear PTAR are strictly less expressive than PTAR.
\end{lemma}

\begin{proof}
Let $\Sigma=\{\sigma^{(2)},\gamma^{(1)},\#^{(0)}\}$ and consider the tree language $L_3=\{\gamma^n(\sigma(\gamma^n\#,\gamma^n\#))\mid n\in\Nat\}$. This tree language can be recognized by a $2$-PTA $\A=(\{q_0,q_1\},\Sigma,q_0,\Delta,C)$ where $C=\{(i,i)\mid i\in\Nat\}$ and $\Delta$ consists of the transitions $q_0\to\gamma(q_0(1,0))$, $q_0\to\sigma(q_1(0,0),q_1(0,0))$, $q_1\to\gamma(q_1(0,1))$, and $q_1\to\#$.

On the other hand, it is not hard to see that $L_3$ cannot be recognized by any linear PTAR $\A'$: by definition, each transition recognizing $\sigma$ in $\A'$ has to be of the form $q\to\sigma(q_1(d),q_2(\reset))$ or $q\to\sigma(q_1(\reset),q_2(d))$. Then the argumentation is very similar to the proof of Lemma \ref{pta-subset-gpta}: we will find a state $p$ such that (1) $(p,\bm{0})$ occurs in a successful computation of $\A'$ and (2) there are $n_1,n_2\in\Nat$ with $n_1\neq n_2$, $(p,\bm{0})\Rightarrow_{\A'}^*\gamma^{n_1}\#$, and $(p,\bm{0})\Rightarrow_{\A'}^*\gamma^{n_2}\#$.  Thus, $\A'$ cannot recognize $L_3$.
\end{proof}





\section{Decidability}

Now we investigate the question of decidability for two basic problems of PTAR -- the non-emptiness problem and the membership problem.\footnote{Note that the universality problem is undecidable already for Parikh string automata \cite[Prop. 7]{KlaRue03}.} The former is undecidable in general: as soon as we consider PTA of at least dimension 3, we can simulate calculations of two-counter machines \cite{Min67,Joh00} in a similar way as it was done in \cite[Lemma 3.4]{LinMitSce92} for \emph{and-branching two-counter machines without zero-test} (ACM).

A \emph{two-counter machine (2CM)} is a tuple $M=(Q,q_0,Q_f,T)$ where $Q$ is a finite set of states, $q_0\in Q$ is the initial state, $Q_f\subseteq Q$ is a set of final states and $T$ is a finite set of transitions of the following two forms:
\[(q,f,q') \tag{instruction}\]
\[(q,p,q') \tag{zero-test}\]
where $q,q'\in Q$, $p\in\{0(1),0(2)\}$, and $f\in\{\inc(1),\inc(2),\dec(1),\dec(2)\}$.

Define for convenience of notation $\neg 1=2$ and $\neg 2=1$. For each $\tau\in T$ we let $\Rightarrow^\tau$ be the binary relation on $Q\times\Nat\times\Nat$ such that for each $(q,k_1,k_2),(q',k_1',k_2')\in (Q\times\Nat\times\Nat)$ we have $(q,k_1,k_2)\Rightarrow^\tau(q',k_1',k_2')$ if either
\begin{itemize}
    \item $\tau=(q,\inc(i),q')$ for some $i\in\{1,2\}$, $k_i'=k_i+1$, and $k_{\neg i}'=k_{\neg i}$, or
    \item $\tau=(q,\dec(i),q')$ for some $i\in\{1,2\}$, $k_i>0$, $k_i'=k_i-1$, and $k_{\neg i}'=k_{\neg i}$, or
    \item $\tau=(q,0(i),q')$, $k_i=0$, $k_1'=k_1$, and $k_2=k_2'$.
\end{itemize}
We let $\Rightarrow_M=\bigcup_{\tau\in T}\Rightarrow^\tau$. We say that $M$ \emph{accepts} if $(q_0,0,0)\Rightarrow_M^*(q_f,0,0)$ for some $q_f\in Q_f$. By the classical result that, given a 2CM $M$ and $k_1,k_2\in\Nat$, it is undecidable whether $(q_0,k_1,k_2)\Rightarrow_M^*(q_f,0,0)$ \cite{Min67}, it is straightforward to obtain undecidability of acceptance of 2CM.

\begin{lemma}[{\cite{Min67}}]\label{lem:2CM}
    Let $M$ be a 2CM. It is undecidable whether $(q_0,0,0)\Rightarrow_M^*(q_f,0,0)$ for some $q_f\in Q_f$.
\end{lemma}

\begin{theorem}\label{thm:undecidable}
    For each $m\geq 3$ and $m$-PTA $\A$ it is undecidable whether $\L(\A)\neq \emptyset$.
\end{theorem}
\begin{proof}
    To prove the statement we reduce the acceptance problem of 2CM to the emptiness problem of PTA similar to the proof of \cite[Lemma 3.4]{LinMitSce92}. The idea is to simulate zero-tests with branching: while in the right successor the calculation continues as if the zero-test had been successful, in the left subtree it is checked whether the zero-test is indeed successful. In contrast to 2CM and ACM, PTAs cannot decrement their counters. Therefore, we need 3 counters to represent the counter values of 2CM -- the counter configuration $(s_1,s_2,l)$ of a PTA stands for the value $(s_1-l,s_2-l)$ of the 2CM; each $(j,j,j)$ represents $(0,0)$. In addition, for each decrement of counter $i$ it must be tested that $l$ is smaller than $s_i$, this also happens via branching.

    Given a 2CM $M=(Q,q_0,Q_f,T)$, we construct the 3-PTA $\A$ as follows: Let $\Sigma=\{\sigma^{(2)},\gamma^{(1)},\alpha^{(0)}\}$ be a ranked alphabet and $\A=(Q',\Sigma,q_0,\Delta, C)$ where $Q'=Q\cup\{=_1,=_2,<_1,<_2\}$, $C=\{(i,i,i)\mid i\in\Nat\}$, and $\Delta$ consists of the following transitions:
    \begin{itemize}
        \item for each transition of the form $(q,\inc(i),q')$ in $T$, the transition $q\to\gamma(q'(d))$ is in $\Delta$ where $d=(2,1,1)$ if $i=1$ and $d=(1,2,1)$ if $i=2$,
        \item for each transition of the form $(q,\dec(i),q')$ in $T$, the transition \[q\to\sigma(<_i(d),q'(d))\] is in $\Delta$ where $d=(0,1,1)$ if $i=1$ and $d=(1,0,1)$ if $i=2$,
        \item for each transition of the form \mbox{$(q,0(i),q')$} in $T$, the transition \[q\to\sigma(=_i(0,0,0),q'(0,0,0))\] is in $\Delta$,
        \item for each $q_f\in Q_f$ the transition $q_f\to\alpha$ is in $\Delta$,
        \item for each $d\in\{(0,1,0),(0,0,1),(1,0,1)\}$ the transition $<_1\to\gamma(<_1(d))$ is in $\Delta$ and for each $d'\in\{(1,0,0),(0,0,1),(0,1,1)\}$ the transition $<_2\to\gamma(<_2(d'))$ is in $\Delta$,
        \item for each $d\in\{(1,0,1),(0,1,0)\}$ the transition $=_1\to\gamma(=_1(d))$ is in $\Delta$ and for each $d'\in\{(0,1,1),(1,0,0)\}$ the transition $=_2\to\gamma(=_2(d'))$ is in $\Delta$, and
        \item the transitions $<_i\to \alpha$ and $=_i\to \alpha$ are in $\Delta$ for each $i\in\{1,2\}$.
    \end{itemize}

    We can show that $\L(\A)\neq\emptyset$ if and only if $(q_0,0,0)\Rightarrow^*_M(q_f,0,0)$ for some $q_f\in Q_f$ by induction on the length of the respective computations. For this, we note that the mapping $\varphi\colon T\to\Delta$ given by the above construction is an injection. Moreover, the following three observations are helpful:

\begin{observation}\label{obs1}
    Let $s_1,s_2,l\in\Nat$. Then $(<_i,(s_1,s_2,l))\Rightarrow^*_\A\gamma^n(<_i,(j,j,j))\Rightarrow_\A\gamma^n(\alpha)$ for some $j\in\Nat$ if and only of either
    $i=1$ and $l\leq s_1$ or $i=2$ and $l\leq s_2$.
\end{observation}

\begin{observation}\label{obs2}
    Let $s_1,s_2,l\in\Nat$. Then $(=_i,(s_1,s_2,l))\Rightarrow^*_\A\gamma^n(=_i,(j,j,j))\Rightarrow_\A\gamma^n(\alpha)$ for some $j\in\Nat$ if and only of either
    $i=1$ and $s_1= l$ or $i=2$ and $s_2= l$.
\end{observation}

\begin{observation}\label{obs3}
    Let $q_1,q_2\in Q$, let $s_1,s_2,l,s_1',s_2',l'\in\Nat$, and let $\zeta\in C_\Sigma(X_1)$. If $(q_1,(s_1,s_2,l))\Rightarrow^*_\A\zeta[(q_2,(s_1',s_2',l'))]$, then also $(q_1,(s_1+1,s_2+1,l+1))\Rightarrow^*_\A\zeta[(q_2,(s_1'+1,s_2'+1,l'+1))]$.
\end{observation}
By using Lemma \ref{lem:2CM}, we can conclude that non-emptiness of $m$-PTA for $m\geq 3$ is undecidable.
\end{proof}


Now we come to a case of PTAR for which the situation is different: we can show that for linear PTAR non-emptiness is decidable. To do so, we use the fact that for every non-empty linear PTAR there must be a tree with less than $|Q|+1$ non-reset paths (Lemma \ref{lem:spinal-bounded}) and, thus, reduce the problem to non-emptiness of Parikh string automata, which is decidable (Lemma \ref{pa-dec}).

\begin{definition}
    Let $\A=(Q,\Sigma,q_0,\Delta,C)$ be a linear PTAR, $U\subseteq Q$, and $q\in Q$. The \emph{$(U,q)$-linearization automaton of $\A$} is the PA $\A'=(Q,\Sigma,q,\Delta',F,C)$ where \[F=\{p\mid p\to\alpha\in\Delta,\alpha\in\Sigma^{(0)}\}\cup\{p\mid p\to\sigma(q_1(d_1),\ldots,q_n(d_n))\in\Delta, d_1,\ldots,d_n=\reset, q_1,\ldots,q_n\in U\}\] and $\Delta'$ contains the transition $(p,\sigma,d,p')$ if and only if there is a transition $p\to\sigma(q_1(d_1),\ldots,q_n(d_n))\in\Delta$, $i\in[n]$, such that $q_i=p'$ and $d_i=d\neq\reset$, and $q_1,\ldots,q_{i-1},q_{i+1},\ldots,q_n\in U$.
\end{definition}

\renewcommand{\algorithmicrequire}{\textbf{Input:}}
\renewcommand{\algorithmicensure}{\textbf{Output:}}
\begin{algorithm}
\label{alg:nonemptiness}
\caption{Decision procedure for non-emptiness of linear PTAR}
\begin{algorithmic}
\Require linear PTAR \(\A = (Q,\Sigma,q_0,\Delta,C)\)
\Ensure ``\(\L(\A) \ne \emptyset\)'' if \(\L(\A) \ne \emptyset\), otherwise ``\(\L(\A) = \emptyset\)''.
\State $U_0 \gets\emptyset$, \(i \gets 0\)
\Repeat
  \State \(U_{i+1} \gets U_i\)
  \For{\(q \in Q\)}
    \State Construct the \((U_i, q)\)-linearization automaton \(\A'\) of \(\A\).
    \If{\(\L(\A') \ne \emptyset\)}
      \State \(U_{i+1} \gets U_{i+1} \cup \{q\}\)
    \EndIf
  \EndFor
  \State \(i \gets i + 1\)
\Until{\(U_{i} = U_{i-1}\)}
 \State \textbf{if} {\(q_0 \in U_i\)} \textbf{then} Output ``\(\L(\A) \ne \emptyset\)'' \textbf{else} Output ``\(\L(\A) = \emptyset\)'' \textbf{end if}
\end{algorithmic}
\end{algorithm}

\begin{theorem}\label{thm:lin-empt}
    Given a linear PTAR $\A$, it is decidable whether $\L(\A)\neq\emptyset$.
\end{theorem}
\begin{proof}
  Consider Algorithm~1.  We claim that this algorithm is a decision procedure for the non-emptiness problem of linear PTAR.

Observe that, for every \(U \subseteq Q\) and every \(q \in Q\), when \(\A'\) is the \((U,q)\)-linearization automaton of \(\A\), we have
\[\L(\A') \ne \emptyset \qquad \text{iff} \qquad \exists s \in S_q\colon\,\stateseq(s) \in U^*\,\text{.} \tag{\(\star\)}\]

By this property, we can conclude that the following loop invariant holds for the outer loop of the algorithm: for every \(j \in \Nat\) and \(q \in Q\), we have
\[q \in U_j \qquad \text{iff} \qquad \exists d \in D_q \colon\, \height(d) \leq j\,\text{.}\]

The proof of the loop invariant is by induction on \(j\).  The base case \(j = 0\) is vacuously true, so assume the property is proven for some \(j \in \Nat\).  For the direction ``only if'', assume some \(q \in U_{j+1}\).
The case \(q \in U_j\) is already covered by the induction hypothesis, so it remains to consider \(q \in U_{j+1} \setminus U_j\).
Then the language of the \((U_j,q)\)-linearization automaton of \(\A\) is nonempty, hence there is some \(s \in S_q\) with \(\stateseq(s) \in U_j^*\) by (\(\star\)).  Moreover, by the induction hypothesis, there are, for all states \(q_i\) in \(\stateseq(s)\), computation trees \(d_i \in D_{q_i}\) with \(\height(d_i)\leq j\).  So the computation tree \(s(d_1, \ldots, d_n) \in D_q\) has height at most \(j + 1\).

For the direction ``if'', assume that there is some \(d \in D_q\) such that \(\height(d) \leq j+1\). Again, if \(\height(d) < j+1\), we can apply the induction hypothesis and are done immediately.  So consider the case \(\height(d) = j+1\).
Then \(d = s(d_1, \ldots, s_\ell)\) for some \(s \in S_q\), \(\ell \in \Nat\) and \(d_i \in D_{q_i}\) for all \(i \in [\ell]\).  In particular \(\height(d_i) \leq j\), so \(q_i \in U_j\) by the induction hypothesis.
Moreover, by (\(\star\)), the language of the \((U_{j}, q)\)-linearization automaton is nonempty, so we obtain that \(q \in U_{j+1}\).

To show correctness of the algorithm, assume that the algorithm outputs ``\(\L(\A) \ne \emptyset\)''.  Then we have \(q_0 \in U_k\), where \(k\) is the value of the counter \(i\) after the outer loop terminates.  By the loop invariant, there is some \(d \in D_{q_0}\), and by Lemma~\ref{lem:semantics}, \(\L(\A) \ne \emptyset\).

For the proof of completeness of the algorithm, assume that \(\L(\A) \ne \emptyset\). By Lemma~\ref{lem:spinal-bounded}, there is some \(d \in D_{q_0}\) with \(\height(d) \leq \abs{Q}\), and by the loop invariant, this means that \(q_0 \in U_j\) for some \(j \leq \abs{Q}\). Thus, also \(q_0 \in U_k\), where \(k\) is the value of the counter \(i\) after the outer loop terminates.  So the algorithm outputs ``\(\L(\A) \ne \emptyset\)''.
\end{proof}


Finally, we obtain that membership is decidable for arbitrary PTAR.

\begin{theorem}\label{thm:member}
    Given an $m$-PTAR $\A$ over $\Sigma$ and a tree $\xi\in T_\Sigma$, it is decidable whether $\xi\in\L(\A)$.
\end{theorem}

\begin{proof}
In order to check whether $\xi\in\L(\A)$, the naive approach is sufficient: Clearly, $\xi\in\L(\A)$ if and only if $\compA{\xi}\neq\emptyset$. As each $t\in\compA{\xi}$ is of length $|\xi|$, we can simply guess a sequence of transitions $\tau_1\ldots\tau_{|\xi|}$ and check whether its application in lexicographic order results in a valid computation. This mainly involves to ensure a valid state behavior and to test for each leaf whether the reached counter configuration is an element of $C$. The latter is decidable due to Lemma \ref{lemma:dec-C}.
\end{proof}

\section{Conclusion}
In this work, we introduced non-global PTA and compared its expressive power with that of GPTA. To do so, we generalized an exchange lemma known from Parikh word automata to GPTA. Furthermore, we investigated the question of decidability of non-emptiness and membership for PTA and linear PTAR.

\paragraph{Future work}
Our investigations in this paper were only a first step and raise many more questions that can be addressed in future work. In particular, we think it worthwhile to further investigate the following questions:
\begin{itemize}
    \item Is it possible to formulate an exchange lemma for (linear) PTA(R)? Since the successful computations of all subtrees of a node depend on the current counter configuration, our attempts to reorder parts of a computation have not been successful so far.
    \item Are PTAR strictly more expressive than PTA?
\end{itemize}
Finally, we did not investigate closure properties for the different models introduced in this work as well as complexities of their non-emptiness and membership problems. We think that a further study could contribute to a more complete picture.

\paragraph{Acknowledgements} We want to thank the reviewers for their insightful and detailed comments, which helped us to improve the paper. In particular, one of the reviewers had a really nice idea how to strengthen the undecidability result from 4-PTA to 3-PTA.

%
%
%
\bibliographystyle{eptcs}
\bibliography{lit}

\newpage
\appendix

\end{document}